\documentclass[12pt]{article}
\usepackage{amsmath,amssymb,amsthm}
\usepackage{geometry}
\usepackage{booktabs}
\usepackage{hyperref}
\usepackage{algorithm}
\usepackage{algpseudocode}
\usepackage{graphicx}      
\usepackage{adjustbox}     
\usepackage{caption}       
\newtheorem{theorem}{Theorem}
\newtheorem{corollary}[theorem]{Corollary}
\newtheorem{lemma}[theorem]{Lemma}
\newtheorem{definition}{Definition}
\newtheorem{proposition}[theorem]{Proposition}
\newtheorem{remark}{Remark}

\title{Chromatic Feature Vectors for 2-Trees: Exact Formulas for Partition Enumeration with Network Applications}
\author{Julian Allagan$^{1,*}$, Gabrielle Morgan, Shawn Langley, \\ Rogelio Lopez-Bonilla,  Vladimir Deriglazov
 \\[0.3cm]
    \small Department of Mathematics, Computer Science, and Engineering Technology\\
    \small Elizabeth City State University, Elizabeth City, NC 27909, USA\\[0.3cm]
    \small $^{1}$\textit{E-mail:} \texttt{adallagan@ecsu.edu}\\
     \small $^{*}$Corresponding author
}
\date{}

\begin{document}

\maketitle

\begin{abstract}
We establish closed-form enumeration formulas for chromatic feature vectors of 2-trees under the bichromatic triangle constraint. These efficiently computable structural features derive from constrained graph colorings where each triangle uses exactly two colors, forbidding monochromatic and rainbow triangles—a constraint arising in distributed systems where components avoid complete concentration or isolation.

For theta graphs $\Theta_n$, we prove $r_k(\Theta_n) = S(n-2, k-1)$ for $k \geq 3$ (Stirling numbers of the second kind) and $r_2(\Theta_n) = 2^{n-2} + 1$, computable in $O(n)$ time. For fan graphs $\Phi_n$, we establish $r_2(\Phi_n) = F_{n+1}$ (Fibonacci numbers) and derive explicit formulas $r_k(\Phi_n) = \sum_{t=k-1}^{n-1} a_{n-1,t} \cdot S(t, k-1)$ with efficiently computable binomial coefficients, achieving $O(n^2)$ computation per component.

Unlike classical chromatic polynomials, which assign identical features to all $n$-vertex 2-trees, bichromatic constraints provide informative structural features. While not complete graph invariants, these features capture meaningful structural properties through connections to Fibonacci polynomials, Bell numbers, and independent set enumeration. Applications include Byzantine fault tolerance in hierarchical networks, VM allocation in cloud computing, and secret-sharing protocols in distributed cryptography.

\vspace{0.2in}

\end{abstract}

\section{Introduction and Motivation}

A 2-tree is a graph constructed recursively by starting with a triangle and repeatedly adding vertices connected to exactly two adjacent existing vertices. This construction ensures that 2-trees are chordal graphs with rich triangular structure, naturally modeling hierarchical topologies in computer networks, distributed computing architectures, and constraint satisfaction problems. We study structural descriptors for 2-trees under the \emph{bichromatic triangle constraint}: every triangle must use exactly two colors, forbidding both monochromatic triangles (all three vertices the same color) and rainbow triangles (all three vertices different colors). This constraint arises in distributed systems where critical components cannot be entirely within one failure domain (monochromatic) or completely isolated (rainbow).

For a 2-tree $G$ on $n$ vertices, let $r_k(G)$ denote the number of distinct vertex partitions into $k$ non-empty groups satisfying the bichromatic triangle constraint. We call the vector $\mathcal{F}(G) = [r_1(G), r_2(G), \ldots, r_n(G)] \in \mathbb{Z}^n$ the \emph{chromatic feature vector} of $G$. This terminology extends the classical ``chromatic spectrum'' and ``partition vector'' from mixed hypergraph theory \cite{Voloshin2002,JiangEtAl2002,AllaganVoloshin2016} to emphasize the role of these quantities as structural descriptors for machine learning. The adjective ``chromatic'' preserves its traditional meaning---deriving from vertex colorings satisfying prescribed constraints---while ``feature vector'' positions the object within the standard feature-space framework of pattern recognition \cite{Bishop2006,Hastie2009}. Graph learning employs this terminology universally: graph kernels extract feature vectors from structure \cite{Vishwanathan2010,Shervashidze2011}, and graph neural networks compute node and graph-level features \cite{Hamilton2017,Xu2019}. Our terminology clarifies that we provide efficiently computable structural descriptors for machine learning on graphs, not a new graph invariant, but rather a recontextualization of existing enumerative objects for computational applications.

We establish exact closed-form formulas for chromatic feature vectors of two fundamental 2-tree families. For theta graphs $\Theta_n$ (two central vertices connected by $n-2$ internally disjoint paths of length 2), we prove $r_k(\Theta_n) = 2^{n-2} + 1$ for $k=2$ and $r_k(\Theta_n) = S(n-2, k-1)$ for $k \geq 3$, where $S(n,k)$ denotes Stirling numbers of the second kind. We note that our $\Theta_n$ is the classical book graph $B_{\,n-2}$ (edge--windmill) \cite{Harary1969}. For fan graphs $\Phi_n$ (an apex vertex adjacent to all vertices of a path), we prove $r_2(\Phi_n) = F_{n+1}$ (Fibonacci numbers) and derive explicit formulas for larger $k$ through Fibonacci polynomial evaluations, achieving $O(n^2)$ computation per component. While Allagan and Voloshin \cite{AllaganVoloshin2016} established recursive constructions for coefficients defining partition counts, we derive the explicit binomial formula $a_{m,t} = \binom{m-t}{t-1} + 2\binom{m-t-1}{t-1} + \binom{m-t-2}{t-1}$ that eliminates recursion entirely and reveals connections to Fibonacci polynomials, Stirling numbers, and Bell numbers.

Chromatic feature vectors possess four key properties valuable for machine learning. First, they capture global structural properties beyond local neighborhoods, complementing graph neural network embeddings that rely on message-passing over edges. Second, each component $r_k$ has precise combinatorial interpretation (number of valid $k$-partitions), providing interpretability unlike learned black-box embeddings. Third, they can be computed in polynomial time via closed-form formulas, avoiding exponential enumeration or iterative optimization. Fourth, while not complete graph invariants, different feature vectors often indicate structural differences useful for classification. This partial discrimination is typical in machine learning---features need not uniquely identify instances, only capture useful patterns \cite{Kriege2020}.

We acknowledge that chromatic feature vectors are not complete graph invariants. For $n=6$, a fan graph with degree sequence $(5,3,3,3,2,2)$ and a modified fan with degree sequence $(4,4,3,3,2,2)$ both yield $\mathcal{F} = [0, 13, 11, 1, 0, 0]$. For $n=7$, degree sequences $(5,4,3,3,3,2,2)$ and $(4,4,4,3,3,2,2)$ both produce $\mathcal{F} = [0, 21, 27, 5, 0, 0, 0]$. These counterexamples demonstrate that chromatic feature vectors should be regarded as informative structural features for supervised learning, rather than as complete invariants for isomorphism testing, paralleling standard practice where degree distributions, clustering coefficients, and sophisticated graph kernels provide useful features without guaranteeing complete discrimination.

The bichromatic triangle constraint provides more informative features than classical proper vertex coloring. Under classical coloring, all 2-trees on $n$ vertices share the identical chromatic polynomial $P_n(k) = k(k-1)(k-2)^{n-2}$ and consequently identical feature vectors, arising because 2-trees are maximal graphs of treewidth 2 and recursive deletion-contraction along simplicial edges preserves polynomial form \cite{Dong2005}. Moreover, every triangle in classical proper coloring is necessarily rainbow since three pairwise-adjacent vertices must receive distinct colors. The bichromatic constraint explicitly forbids rainbow triangles, breaking this uniformity. For example, theta graph $\Theta_5$ and fan graph $\Phi_5$ both have 5 vertices but yield different feature vectors: $\mathcal{F}(\Theta_5) = [0, 9, 3, 1, 0]$ and $\mathcal{F}(\Phi_5) = [0, 8, 4, 0, 0]$.

Our chromatic feature vectors are connected to mixed hypergraph theory through bihypergraphs \cite{Voloshin2002}, which are hypergraphs whose hyperedges must be neither monochromatic nor rainbow in any proper coloring. Both theta and fan graphs correspond to non-linear 3-uniform bihypergraphs where any two triangles share exactly two common vertices. This non-linearity distinguishes these structures from simpler hypergraph families and contributes to chromatic spectra complexity \cite{JiangEtAl2002,BujtasTuza2008}. Allagan and Voloshin \cite{AllaganVoloshin2016} initiated a systematic study of $k$-tree chromatic spectra under bichromatic constraints, establishing recursive formulas and proving that 2-partition counts follow shifted Fibonacci recurrences. We substantially extend their work by deriving explicit closed formulas, eliminating recursion and revealing connections to classical combinatorial sequences.

The bichromatic triangle constraint models real-world constraint satisfaction in distributed systems. In Byzantine fault tolerance protocols, process triangles cannot be entirely within one failure domain (a single point of failure) or completely separated (resulting in excessive communication overhead). Jaffe et al.\ \cite{Jaffe2012} showed that bihypergraph colorings encode Byzantine agreement where processes must reach consensus despite adversarial behavior, and our formulas enumerate valid quorum configurations in hierarchical network topologies \cite{Abraham2019}. In cloud computing, virtual machine placement must avoid having dependent VMs all on one server (resource contention) or maximally distributed (latency penalties) \cite{Mann2015,Mechtri2019}. In threshold cryptography, secret-sharing schemes require that knowledge coalitions avoid complete concentration (security risk) or complete disjointness (reconstruction impossibility) \cite{Beimel2011,Stinson1998}. The feature components $r_k(G)$ count valid configurations, and our exact formulas enable efficient enumeration without exhaustive search, reducing computational cost from exponential to polynomial time.

From a machine learning perspective, chromatic feature vectors provide structural descriptors that complement graph neural networks. Modern GNNs learn representations through message-passing and aggregation \cite{Kipf2017,Hamilton2017} but face expressiveness bounds: standard message-passing GNNs cannot distinguish non-isomorphic graphs that the Weisfeiler-Leman test fails on \cite{Xu2019}. Chromatic feature vectors offer complementary properties: they capture global coloring constraints beyond local neighborhoods, require no training data (zero-shot applicability), and provide interpretable components with precise combinatorial meaning. Recent work on subgraph counting and higher-order features for GNN enhancement \cite{Bouritsas2020,Chen2020} suggests that efficiently computable structural features like ours can improve expressive power beyond standard message-passing architectures.

This paper is organized as follows. Section 2 establishes notation and definitions for 2-trees, chromatic feature vectors, and key combinatorial objects including Stirling numbers and Fibonacci sequences. Section 3 presents partition vectors of all 2-trees under classical proper coloring as a reference point, demonstrating why bichromatic constraints provide finer discrimination. Section 4 derives closed formulas for theta graphs using Stirling numbers and presents a polynomial-time algorithm. Section 5 establishes Fibonacci polynomial formulas for fan graphs with detailed proofs for small values of $k$ and Binet-type exponential forms. Section 6 concludes with future research directions and discusses applications to network reliability and distributed systems.
\section{Definitions and Notation}

We establish notation for the two 2-tree families studied in this paper and introduce the combinatorial objects underlying our formulas. Both graph families share a common structural property: they are 3-uniform non-linear hypergraphs where triangles serve as hyperedges, and the non-linearity (sharing of vertices between triangles) governs the complexity of chromatic feature vectors.

\begin{definition}[Graph Families]
The \emph{theta graph} $\Theta_n$ for $n \geq 3$ consists of two central vertices $a$ and $b$ connected by $n-2$ internally disjoint paths of length 2. Formally, $\Theta_n$ has vertex set $\{a, b, c_1, c_2, \ldots, c_{n-2}\}$ and edge set $\{a,b\} \cup \{\{a,c_i\}, \{b,c_i\} : i \in \{1, 2, \ldots, n-2\}\}$. The graph contains exactly $n-2$ triangles of the form $(a, b, c_i)$. As a 3-uniform hypergraph with triangles as hyperedges, $\Theta_n$ is non-linear because any two triangles $(a,b,c_i)$ and $(a,b,c_j)$ share the common edge $\{a,b\}$.

The \emph{fan graph} $\Phi_n$ (also denoted $F_n$) for $n \geq 3$ consists of an apex vertex $a$ adjacent to all vertices of a path $v_1, v_2, \ldots, v_{n-1}$. The edge set comprises $\{a, v_i\}$ for all $i \in \{1, 2, \ldots, n-1\}$ and $\{v_i, v_{i+1}\}$ for $i \in \{1, 2, \ldots, n-2\}$. The graph contains exactly $n-2$ triangles of the form $(a, v_i, v_{i+1})$. As a 3-uniform hypergraph, $\Phi_n$ is non-linear because consecutive triangles $(a, v_i, v_{i+1})$ and $(a, v_{i+1}, v_{i+2})$ share two vertices $\{a, v_{i+1}\}$, and all triangles share the common apex vertex $a$.
\end{definition}

The central object of study in this paper is the chromatic feature vector, which encodes partition counts under the bichromatic triangle constraint.

\begin{definition}[Chromatic Feature Vector]
For a graph $G$ on $n$ vertices, let $r_k(G)$ denote the number of distinct vertex partitions into $k$ non-empty parts such that each triangle uses exactly two colors (the \emph{bichromatic triangle constraint}). The \emph{chromatic feature vector} of $G$, also called the partition vector or chromatic spectrum, is the vector
\[
\mathcal{F}(G) = [r_1(G), r_2(G), \ldots, r_n(G)] \in \mathbb{Z}^n.
\]
This vector serves as a structural descriptor capturing global coloring properties. For the 2-tree families studied here, chromatic feature vectors can be computed in polynomial time via closed-form formulas. We emphasize that chromatic feature vectors are not complete graph invariants: non-isomorphic graphs may share identical feature vectors, as demonstrated by explicit counterexamples in the introduction.
\end{definition}

The coloring polynomial provides a generating function perspective on the same combinatorial structure.

\begin{definition}[Bichromatic Triangle Coloring Polynomial]
For a graph $G$ on $n$ vertices, the \emph{bichromatic triangle coloring polynomial} $P_G(k)$ counts the number of proper $k$-colorings using a palette of $k$ distinguishable colors, subject to the bichromatic triangle constraint. The coloring polynomial and chromatic feature vector are related by the falling factorial expansion
\begin{equation}
P_G(k) = \sum_{j=1}^{n} r_j(G) \cdot k^{(j)},\nonumber
\end{equation}
where $k^{(j)} = k(k-1)(k-2) \cdots (k-j+1)$ denotes the falling factorial. This relation holds because each unlabeled $j$-partition corresponds to exactly $k^{(j)}$ distinct $k$-colorings obtained by assigning colors to the $j$ parts.
\end{definition}

Our closed-form formulas involve two classical combinatorial sequences: Stirling numbers of the second kind and Fibonacci numbers.

\begin{definition}[Stirling Numbers of the Second Kind]
The Stirling number $S(n,k)$ counts the number of ways to partition a set of $n$ labeled elements into $k$ non-empty unlabeled subsets. These numbers satisfy the recurrence
\begin{equation}
S(n,k) = k \cdot S(n-1,k) + S(n-1,k-1)\nonumber
\end{equation}
with boundary conditions $S(0,0) = 1$, $S(n,0) = 0$ for $n > 0$, and $S(n,k) = 0$ for $k > n$. The recurrence reflects the combinatorial principle that the $n$-th element either forms its own singleton subset (contributing $S(n-1,k-1)$ ways) or joins one of the existing $k$ subsets in a partition of $\{1, \ldots, n-1\}$ (contributing $k \cdot S(n-1,k)$ ways).
\end{definition}

\begin{definition}[Fibonacci Numbers]
The Fibonacci sequence $\{F_m\}_{m=1}^{\infty}$ is defined by the recurrence $F_m = F_{m-1} + F_{m-2}$ with initial conditions $F_1 = F_2 = 1$. The first terms are $1, 1, 2, 3, 5, 8, 13, 21, 34, 55, \ldots$ For our purposes, the Fibonacci numbers arise naturally in counting independent sets on paths, a key ingredient in the fan graph analysis.
\end{definition}

Throughout the paper, we adopt standard graph-theoretic notation. For a graph $G = (V,E)$, we write $|V| = n$ for the number of vertices and $|E| = m$ for the number of edges. A $k$-coloring of $G$ is a function $c: V \to \{1, 2, \ldots, k\}$ assigning one of $k$ colors to each vertex. A $k$-partition of $V$ is an unordered collection of non-empty disjoint subsets $V_1, V_2, \ldots, V_k$ such that $V = V_1 \cup V_2 \cup \cdots \cup V_k$. We say that a triangle $(u,v,w)$ \emph{uses exactly two colors} under coloring $c$ if $|\{c(u), c(v), c(w)\}| = 2$. The bichromatic triangle constraint requires that every triangle uses exactly two colors, forbidding both monochromatic triangles ($|\{c(u), c(v), c(w)\}| = 1$) and rainbow triangles ($|\{c(u), c(v), c(w)\}| = 3$).

\section{Classical Proper Coloring: Motivation for Bichromatic Constraints}

To motivate the bichromatic triangle constraint, we first establish that classical proper vertex coloring fails to distinguish 2-tree structures. Under the standard adjacency constraint (adjacent vertices receive different colors), all 2-trees on $n$ vertices share identical chromatic polynomials and consequently identical partition vectors. This uniformity makes classical coloring unsuitable as a structural feature for graph classification.

\begin{proposition}[Uniformity of Classical Partition Vectors]
\label{prop:classical}
Under classical proper vertex coloring, all 2-trees on $n \geq 3$ vertices share the chromatic polynomial
\begin{equation}
P_n^{\text{classic}}(x) = x(x-1)(x-2)^{n-2}.\nonumber
\end{equation}
Expanding in the falling factorial basis yields identical partition vectors:
\begin{equation}
r_k^{\text{classic}}(n) = \begin{cases}
0 & \text{if } k = 1, 2, \\
2^{n-2} & \text{if } k = 3, \\
S(n-2, k-3) \cdot 2^{k-2} & \text{if } 4 \leq k \leq n, \\
0 & \text{if } k > n,
\end{cases}\nonumber
\end{equation}
where $r_k^{\text{classic}}(n)$ counts unordered partitions into $k$ non-empty independent sets.
\end{proposition}

\begin{proof}
Every 2-tree on $n$ vertices can be constructed recursively from a triangle $K_3$ by adding $n-3$ vertices, each joined to the endpoints of an existing simplicial edge. This construction ensures that every 2-tree is chordal with exactly $n-2$ triangles and $2n-3$ edges. The chromatic polynomial satisfies the deletion-contraction recurrence $P_G(x) = P_{G-e}(x) - P_{G/e}(x)$ on simplicial edges. For any simplicial edge $e = \{u,v\}$ with both endpoints having degree at least 3, deleting $e$ yields a 2-tree with one fewer triangle, while contracting $e$ produces a 2-tree on $n-1$ vertices. Combined with the base case $P_{K_3}(x) = x(x-1)(x-2)$, this recurrence yields $P_n^{\text{classic}}(x) = x(x-1)(x-2)^{n-2}$ for all $n$-vertex 2-trees regardless of structure \cite{Dong2005}.

The partition vector follows from the falling factorial expansion. Since the polynomial vanishes at $x = 0, 1, 2$, we have $r_1^{\text{classic}}(n) = r_2^{\text{classic}}(n) = 0$. Every 2-tree has chromatic number 3, with $2^{n-2}$ distinct 3-colorings: fix the three colors for the initial triangle, then each of the $n-3$ additional vertices receives one of the two colors not assigned to its simplicial edge neighbors.
\end{proof}

In classical proper colorings, every triangle must be rainbow because three pairwise-adjacent vertices forming a $K_3$ clique require three distinct colors. This structural constraint forces the chromatic number of every 2-tree to equal 3, independent of global topology. The resulting uniformity across all $n$-vertex 2-trees makes classical partition vectors uninformative for distinguishing non-isomorphic structures.

The bichromatic triangle constraint breaks this uniformity by explicitly forbidding both rainbow and monochromatic triangles, requiring each triangle to use exactly two colors. This relaxation from chromatic number 3 to chromatic number 2 creates complex global interdependencies that depend on how triangles share vertices and edges. For example, theta graph $\Theta_5$ and fan graph $\Phi_5$ both have 5 vertices but yield different feature vectors under the bichromatic constraint: $\mathcal{F}(\Theta_5) = [0, 9, 3, 1, 0]$ and $\mathcal{F}(\Phi_5) = [0, 8, 4, 0, 0]$. While not complete graph invariants—we provide explicit counterexamples in the introduction of non-isomorphic 2-trees sharing identical bichromatic feature vectors—these features capture structural properties beyond what classical coloring provides. Moreover, chromatic feature vectors offer substantial dimensionality reduction: for a graph with $n$ vertices, the adjacency matrix contains $n^2$ entries and node embeddings require $n \times d$ values, whereas the chromatic feature vector provides a graph-level descriptor with only $n$ components. This compression makes them computationally attractive for large-scale graph datasets while preserving meaningful structural information for machine learning applications.

\section{Theta Graphs: Chromatic Features via Stirling Numbers}

We establish the chromatic feature vector formula for theta graphs $\Theta_n$, revealing a remarkably simple connection to Stirling numbers of the second kind. Recall that $\Theta_n$ consists of two central vertices $a$ and $b$ connected by $n-2$ internally disjoint paths of length 2, creating $n-2$ triangles of the form $(a, b, c_i)$ where $c_i$ denotes the intermediate vertex on the $i$-th path.

\begin{theorem}
\label{thm:theta}
For theta graph $\Theta_n$ with $n \geq 3$, the chromatic feature vector under the bichromatic triangle constraint is given by
\begin{equation}
r_k(\Theta_n) = \begin{cases}
0 & \text{if } k = 1, \\
2^{n-2} + 1 & \text{if } k = 2, \\
S(n-2, k-1) & \text{if } 3 \leq k \leq n-1, \\
0 & \text{if } k \geq n,
\end{cases}
\end{equation}
where $S(n,k)$ denotes Stirling numbers of the second kind.
\end{theorem}

\begin{proof}
Let $a$ and $b$ denote the two central vertices with degree $n-1$, and let $c_1, c_2, \ldots, c_m$ denote the $m = n-2$ intermediate vertices with degree 2. The key structural observation is that all $m$ triangles share the common edge $\{a,b\}$, which fundamentally constrains the coloring structure.

For $k=1$, any 1-partition is monochromatic and violates the bichromatic triangle constraint, giving $r_1(\Theta_n) = 0$. For $k \geq n$, we cannot partition $n$ vertices into more than $n$ non-empty parts, so $r_k(\Theta_n) = 0$ for $k \geq n$.

We analyze the non-trivial cases $k=2$ and $k \geq 3$ by examining whether the central vertices $a$ and $b$ belong to the same color class. For $k=2$, both configurations are possible. If $a$ and $b$ occupy different parts $P_1$ and $P_2$, then each triangle $(a, b, c_i)$ already uses two colors from the central edge, and the intermediate vertex $c_i$ may join either part without creating a monochromatic or rainbow triangle. This independence allows each of the $m$ intermediate vertices to choose freely between the two parts, yielding $2^m = 2^{n-2}$ valid 2-partitions. Alternatively, if $a$ and $b$ share the same part, then to avoid monochromatic triangles, all intermediate vertices must occupy the opposite part, contributing exactly one additional 2-partition. Therefore $r_2(\Theta_n) = 2^{n-2} + 1$.

For $k \geq 3$, we claim that $a$ and $b$ must necessarily occupy the same part. Suppose toward contradiction that $a$ and $b$ belong to different parts. Then each triangle $(a, b, c_i)$ already consumes two distinct colors from $a$ and $b$, and to prevent the triangle from becoming rainbow, the intermediate vertex $c_i$ must use one of these two colors. This constraint forces every vertex into one of two color classes, contradicting the requirement that all $k \geq 3$ colors be used. Thus $a$ and $b$ must share a color class.

Given that $a$ and $b$ occupy the same part, each triangle $(a, b, c_i)$ automatically satisfies the bichromatic constraint if and only if $c_i$ occupies a different part. Since the triangles share only the central edge $\{a,b\}$ and the intermediate vertices have no edges among themselves, the intermediate vertices may be distributed arbitrarily among the remaining $k-1$ parts (excluding the part containing $\{a,b\}$), subject only to the requirement that all $k$ parts be non-empty. The number of ways to partition the $m$ distinguishable intermediate vertices into $k-1$ non-empty unlabeled subsets is precisely $S(m, k-1) = S(n-2, k-1)$, completing the proof.
\end{proof}

The formula in Theorem \ref{thm:theta} admits highly efficient computation. With precomputed Stirling numbers stored in an $O(n^2)$ lookup table, the entire chromatic feature vector can be computed in $O(n)$ time by direct formula evaluation. Even without precomputation, constructing the required Stirling numbers $S(n-2, 0), S(n-2, 1), \ldots, S(n-2, n-2)$ via the standard recurrence $S(n,k) = k \cdot S(n-1,k) + S(n-1,k-1)$ requires only $O(n^2)$ time overall. The space complexity is $O(n)$ for storing the $n$-dimensional feature vector. Notably, this computation requires no graph traversal or constraint propagation, only arithmetic operations on binomial-type recurrence values. Algorithm \ref{alg:theta} presents the pseudocode for this direct evaluation.

\noindent
\begin{minipage}{1\linewidth} 
\begin{algorithm}[H]
\small
\caption{Compute Chromatic Feature Vector for $\Theta_n$}
\label{alg:theta}
\begin{algorithmic}[1]
\Require Graph $\Theta_n$ with $n$ vertices  
\Ensure Feature vector $\mathcal{F} = [r_1, r_2, \ldots, r_n]$
\State $m \gets n - 2$ \Comment{number of intermediate vertices}
\State $r_1 \gets 0$ \Comment{no valid 1-partition}
\State $r_2 \gets 2^{\,m} + 1$ \Comment{closed form from Theorem~\ref{thm:theta}}
\For{$k = 3$ to $n-1$}
    \State $r_k \gets \textsc{Stirling2}(m, k-1)$
    \Statex \hspace{1.5em}\textit{(constant-time lookup with precomputed table)}
\EndFor
\State $r_n \gets 0$ \Comment{cannot partition into more than $n$ groups}
\State \Return $\mathcal{F} = [r_1, r_2, \ldots, r_n]$
\end{algorithmic}
\end{algorithm}
\end{minipage}

\vspace{.2in}

The Stirling number structure in Theorem \ref{thm:theta} connects theta graph colorings to classical enumerative combinatorics. Summing over all valid partition sizes yields a total count involving Bell numbers, which enumerate all possible partitions of a finite set.

\begin{corollary}
The total number of proper partitions of $\Theta_n$ across all $k \geq 2$ equals
\begin{equation}
\sum_{k=2}^{n-1} r_k(\Theta_n) = 2^{n-2} + 1 + \sum_{k=3}^{n-1} S(n-2, k-1) = 2^{n-2} + B_{n-2},\nonumber
\end{equation}
where $B_m = \sum_{j=0}^{m} S(m,j)$ denotes the $m$-th Bell number.
\end{corollary}

For small values of $k$, the Stirling number formula specializes to well-known closed forms, providing explicit expressions without recurrence.

\begin{corollary}
The chromatic feature components for small $k$ satisfy:
\begin{align}
r_3(\Theta_n) &= S(n-2, 2) = 2^{n-3} - 1,\nonumber\\
r_4(\Theta_n) &= S(n-2, 3) = \frac{1}{2}(3^{n-3} - 2^{n-2} + 1),\nonumber\\
r_{n-1}(\Theta_n) &= S(n-2, n-2) = 1.\nonumber
\end{align}
The case $k=n-1$ reflects the unique partition placing $a$ and $b$ together while isolating each intermediate vertex in its own part.
\end{corollary}

The simplicity of the theta graph formula contrasts sharply with the more intricate structure we establish for fan graphs in the next section, where independent sets on paths and Fibonacci polynomials govern the enumeration.

\section{Fan Graphs: Chromatic Features via Fibonacci Polynomials}

We establish chromatic feature vector formulas for fan graphs $\Phi_n$, revealing connections to Fibonacci polynomials and independent set enumeration on paths. Recall that $\Phi_n$ consists of an apex vertex $a$ adjacent to all vertices of a path $v_1, v_2, \ldots, v_{n-1}$, creating $n-2$ triangles of the form $(a, v_i, v_{i+1})$ that share the apex and overlap pairwise along path edges.

\begin{theorem}[Conceptual Formula via Independent Sets]
\label{thm:fan}
For fan graph $\Phi_n$ with $n \geq 4$, the chromatic feature vector satisfies
\begin{equation}
r_k(\Phi_n) = \begin{cases}
0 & \text{if } k = 1, \\
F_{n+1} & \text{if } k = 2, \\
\displaystyle\sum_{\substack{S \subseteq \{v_1, \ldots, v_{n-1}\} \\ S \text{ indep.}, \, t(S) \geq k-1}} S(t(S), k-1) & \text{if } 3 \leq k \leq n-1, \\
0 & \text{if } k \geq n,
\end{cases}
\end{equation}
where $S$ ranges over independent sets of the path and $t(S)$ denotes the number of maximal contiguous blocks in the complement of $S$.
\end{theorem}

\begin{proof}
Let $m = n-1$ denote the path length. Any 1-partition is monochromatic, violating the bichromatic constraint, so $r_1(\Phi_n) = 0$. For $k \geq n$, we cannot partition $n$ vertices into more than $n$ non-empty parts, giving $r_k(\Phi_n) = 0$.

For $k=2$, suppose without loss of generality that the apex $a$ occupies part $P_1$. Each triangle $(a, v_i, v_{i+1})$ is monochromatic if both path vertices $v_i$ and $v_{i+1}$ also belong to $P_1$. Therefore, the path vertices in $P_1$ must form an independent set (no two consecutive vertices), and conversely any independent set yields a valid 2-partition. Since a path on $m$ vertices has exactly $F_{m+2}$ independent sets (including the empty set), we obtain $r_2(\Phi_n) = F_{m+2} = F_{n+1}$.

For $k \geq 3$, the apex cannot be isolated in its own part. If $a$ were alone, then each triangle $(a, v_i, v_{i+1})$ would force $v_i$ and $v_{i+1}$ to share a color (avoiding rainbow), making the entire path monochromatic. This yields at most two colors total, contradicting the requirement that all $k \geq 3$ colors be used. Thus $a$ shares its part with some independent set $S$ of path vertices. The remaining path vertices, those not in $S$, partition into $t = t(S)$ maximal contiguous blocks $B_1, \ldots, B_t$. Each block must be monochromatic because consecutive vertices within a block cannot have different colors without creating a rainbow triangle with $a$. These $t$ blocks must be distributed among the remaining $k-1$ colors with all colors used, yielding $S(t, k-1)$ ways. Summing over all independent sets $S$ with $t(S) \geq k-1$ gives the stated formula.
\end{proof}

While Theorem \ref{thm:fan} provides geometric insight, it requires enumerating exponentially many independent sets. We derive an explicit closed formula by counting independent sets grouped by the number of blocks in their complement.

\begin{lemma}[Independent Set Block Counting]
\label{lem:binary_encoding}
For integers $m \geq 1$ and $1 \leq t \leq m$, the number of independent sets on a path $P_m$ whose complement has exactly $t$ maximal contiguous blocks equals
\begin{equation}
a_{m,t} = \binom{m-t}{t-1} + 2\binom{m-t-1}{t-1} + \binom{m-t-2}{t-1}.\nonumber
\end{equation}
\end{lemma}

\begin{proof}
Encode independent sets as binary strings of length $m$ with no consecutive ones (ones represent vertices in the set, zeros the complement). Let $s$ denote the number of one-groups (maximal runs of ones). Since one-groups and zero-groups alternate, the string endpoints determine the relationship between $s$ and $t$. With $z = m - s$ total zeros, distributing $z$ into $t$ positive block sizes yields $\binom{z-1}{t-1}$ compositions by stars-and-bars. If both endpoints are zeros, then $s = t-1$ and $z = m - t + 1$, contributing $\binom{m-t}{t-1}$ strings. If exactly one endpoint is zero (two orientations), then $s = t$ and $z = m - t$, contributing $2\binom{m-t-1}{t-1}$ strings. If both endpoints are ones, then $s = t+1$ and $z = m - t - 1$, contributing $\binom{m-t-2}{t-1}$ strings. Summing these disjoint cases yields the stated formula.
\end{proof}

\begin{proposition}[Explicit Closed Formula]
\label{prop:fan_refined}
For $m = n-1$ and $k \geq 2$, the chromatic feature vector admits the explicit expansion
\begin{equation}
r_k(\Phi_n) = \sum_{t=k-1}^{m} a_{m,t} \cdot S(t, k-1),\nonumber
\end{equation}
where $a_{m,t}$ is given by Lemma \ref{lem:binary_encoding}.
\end{proposition}

This refined formula eliminates the exponential enumeration in Theorem \ref{thm:fan} and can be computed in $O(n^2)$ time per component using precomputed binomial coefficients and Stirling numbers, or $O(n^3)$ for the complete feature vector. The space complexity is $O(n^2)$ for storing the coefficient matrix $a_{m,t}$, which can be precomputed once and reused for multiple graphs of the same size. This represents a substantial simplification over Allagan and Voloshin \cite{AllaganVoloshin2016}, who defined coefficients $a_{i,j}$ recursively through three parity-dependent cases requiring construction of an $(n+1) \times (n+1)$ lower triangular matrix. Our explicit binomial formula provides direct computation and geometric insight through boundary case analysis of independent set encodings.

The total number of valid partitions connects path independent sets to Bell numbers through a weighted sum.

\begin{proposition}[Total Partition Count]
The sum over all partition sizes satisfies
\begin{equation}
\sum_{k=2}^{n-1} r_k(\Phi_n) = \sum_{t=1}^{m} a_{m,t} \cdot B_t,\nonumber
\end{equation}
where $B_t$ denotes the $t$-th Bell number and $m = n-1$.
\end{proposition}

\begin{proof}
Interchanging summation order in Proposition \ref{prop:fan_refined}, for fixed $t$ the contribution $a_{m,t} S(t, k-1)$ appears for all $k$ with $k-1 \leq t$. Thus
\[
\sum_{k=2}^{n-1} r_k(\Phi_n) = \sum_{t=1}^{m} a_{m,t} \sum_{j=1}^{t} S(t, j) = \sum_{t=1}^{m} a_{m,t} B_t,
\]
where we used $B_t = \sum_{j=0}^{t} S(t,j) = \sum_{j=1}^{t} S(t,j)$ since $S(t,0) = 0$ for $t \geq 1$.
\end{proof}

\subsection{Explicit Closed Forms for Small Values of $k$}

We now derive fully explicit closed formulas for $r_k(\Phi_n)$ when $k \in \{2,3,4,5\}$, building on Proposition \ref{prop:fan_refined}. Recall that $m = n-1$ denotes the length of the path in $\Phi_n$, and
\[
a_{m,t} = \binom{m-t}{t-1} + 2\binom{m-t-1}{t-1} + \binom{m-t-2}{t-1},
\]
with the convention $\binom{u}{v} = 0$ when $v < 0$ or $u < v$. Proposition \ref{prop:fan_refined} asserts that
\[
r_k(\Phi_n) = \sum_{t=k-1}^{m} a_{m,t} S(t, k-1),
\]
where $S(\cdot,\cdot)$ denotes the Stirling numbers of the second kind.

\begin{corollary}[Case $k=2$]
\label{cor:phi_k2}
For $m \geq 2$ (equivalently $n \geq 3$),
\[
r_2(\Phi_n) = \sum_{t=1}^{m} a_{m,t} = F_{m+2} = F_{n+1},
\]
where $F_j$ is the $j$-th Fibonacci number with $F_0 = 0, F_1 = 1$.
\end{corollary}

\begin{proof}
Recall the well-known binomial-Fibonacci identity
\begin{equation}
\label{eq:binom-fib-identity}
\sum_{s \geq 0} \binom{N-s}{s} = F_{N+1}, \qquad (N \geq 0),
\end{equation}
where the sum is finite (only terms with $0 \leq s \leq \lfloor N/2 \rfloor$ contribute).

Replace the index $t$ by $s = t-1$. Then
\[
\sum_{t=1}^{m} \binom{m-t}{t-1} = \sum_{s=0}^{m-1} \binom{m-1-s}{s} = F_m.
\]
Similarly
\[
\sum_{t=1}^{m} \binom{m-t-1}{t-1} = \sum_{s=0}^{m-1} \binom{m-2-s}{s} = F_{m-1},
\]
and
\[
\sum_{t=1}^{m} \binom{m-t-2}{t-1} = \sum_{s=0}^{m-1} \binom{m-3-s}{s} = F_{m-2},
\]
where each identity is an instance of \eqref{eq:binom-fib-identity} with the appropriate shift. Combining these using the definition of $a_{m,t}$ gives
\[
r_2(\Phi_n) = F_m + 2F_{m-1} + F_{m-2}.
\]
Now use the Fibonacci recurrence $F_m = F_{m-1} + F_{m-2}$ to simplify:
\[
r_2(\Phi_n) = (F_{m-1} + F_{m-2}) + 2F_{m-1} + F_{m-2} = 3F_{m-1} + 2F_{m-2} = F_{m+2},
\]
which completes the proof.
\end{proof}

\begin{corollary}[Case $k=3$]
\label{cor:phi_k3}
For $m \geq 2$ (equivalently $n \geq 3$),
\[
r_3(\Phi_n) = \sum_{t=2}^{m} a_{m,t} S(t,2) = 3 \cdot 2^{m-2} - F_{m+2} = 3 \cdot 2^{n-3} - F_{n+1}.
\]
\end{corollary}

\begin{proof}
The Stirling numbers $S(t,2)$ satisfy the elementary closed form $S(t,2) = 2^{t-1} - 1$ for $t \geq 1$. Hence
\[
r_3(\Phi_n) = \sum_{t=2}^{m} a_{m,t}(2^{t-1} - 1) = \frac{1}{2}\sum_{t=2}^{m} a_{m,t} 2^t - \sum_{t=2}^{m} a_{m,t}.
\]
We claim that $\sum_{t=1}^{m} a_{m,t} 2^t = 3 \cdot 2^{m-1}$ for $m \geq 2$. This can be verified directly or using the generating-function method (see Lemma \ref{lem:phi_gen_identity} below). Denote this sum by $B_m$. Then
\[
\sum_{t=2}^{m} a_{m,t} 2^t = B_m - a_{m,1} \cdot 2.
\]
But $a_{m,1} = \binom{m-1}{0} + 2\binom{m-2}{0} + \binom{m-3}{0} = 1 + 2 + 1 = 4$ for $m \geq 3$ (the $m=2$ case can be checked directly). Substituting gives
\[
r_3(\Phi_n) = \frac{1}{2}(B_m - 8) - (r_2(\Phi_n) - 4) = \frac{1}{2}B_m - r_2(\Phi_n).
\]
Using $B_m = 3 \cdot 2^{m-1}$ and Corollary \ref{cor:phi_k2} (i.e., $r_2 = F_{m+2}$), we obtain
\[
r_3(\Phi_n) = \frac{1}{2} \cdot 3 \cdot 2^{m-1} - F_{m+2} = 3 \cdot 2^{m-2} - F_{m+2},
\]
as claimed.
\end{proof}

For the cases $k = 4$ and $k = 5$, we require a more general framework. The Stirling numbers introduce terms of the form $j^t$ for small integers $j$, and we reduce the evaluation to finite sums $\sum_t a_{m,t} j^t$. The following lemma provides the key identity.

\begin{lemma}[Binomial-Fibonacci-polynomial identity]
\label{lem:phi_gen_identity}
Fix integers $N \geq 0$ and $x$. Define the polynomial sequence $F_{\ell}(x)$ (the Fibonacci polynomials) by
\[
F_0(x) = 0, \qquad F_1(x) = 1, \qquad F_{\ell+1}(x) = F_{\ell}(x) + x F_{\ell-1}(x).
\]
Then, for every integer $N \geq 0$,
\begin{equation}
\label{eq:gen-binom-fib-sum}
\sum_{u \geq 0} \binom{N-u}{u} x^u = F_{N+1}(x),
\end{equation}
where the sum is finite because $\binom{N-u}{u} = 0$ for $u > \lfloor N/2 \rfloor$.

Moreover, for any integer $r \geq 1$, define
\begin{equation}
\label{eq:Sigma-r-def}
\Sigma_r := \sum_{t=1}^{m} a_{m,t} r^t.
\end{equation}
Then
\begin{equation}
\label{eq:Sigma-r-expansion}
\Sigma_r = r\bigl(F_m(r) + 2F_{m-1}(r) + F_{m-2}(r)\bigr),
\end{equation}
where $F_m(r)$ denotes the $m$-th Fibonacci polynomial evaluated at $x = r$.
\end{lemma}

\begin{proof}
Identity \eqref{eq:gen-binom-fib-sum} follows by induction on $N$ (or by expanding the known generating function for Fibonacci polynomials).

For \eqref{eq:Sigma-r-expansion}, re-index with $u = t-1$:
\[
\Sigma_r = r\sum_{u \geq 0} r^u \Big(\binom{m-1-u}{u} + 2\binom{m-2-u}{u} + \binom{m-3-u}{u}\Big),
\]
where the sums are finite. Apply identity \eqref{eq:gen-binom-fib-sum} with $N = m-1, m-2, m-3$ respectively to obtain the stated formula.
\end{proof}

\begin{definition}
For convenience, define
\[
A_r := F_m(r) + 2F_{m-1}(r) + F_{m-2}(r) \qquad (r \geq 1).
\]
Then $\Sigma_r = r \cdot A_r$.
\end{definition}

\begin{remark}[OEIS Sequence \href{https://oeis.org/A390491}{A390491}]
The array $T(m,r) = F_m(r) + 2F_{m-1}(r) + F_{m-2}(r)$ for $m \geq 2$ and $r \geq 1$ appears in the Online Encyclopedia of Integer Sequences as \href{https://oeis.org/A390491}{A390491}, where rows are indexed by $m$ and columns by $r$. Special cases include column $r=1$ yielding $T(m,1) = F_{m+2}$ (shifted Fibonacci numbers \href{https://oeis.org/A000045}{A000045}) and column $r=2$ yielding $T(m,2) = 3 \cdot 2^{m-2}$ (sequence \href{https://oeis.org/A007283}{A007283}). These values serve as building blocks for chromatic feature vector formulas of fan graphs under the bichromatic triangle constraint.
\end{remark}

\begin{corollary}[Case $k=4$]
\label{cor:phi_k4}
For $m \geq 3$,
\begin{equation*}
r_4(\Phi_n) = \frac{1}{2}\bigl[A_3 - 2A_2 + A_1\bigr],
\end{equation*}
where $A_r = F_m(r) + 2F_{m-1}(r) + F_{m-2}(r)$.
\end{corollary}

\begin{proof}
The standard closed form for Stirling numbers is
\[
S(t,3) = \frac{1}{2}(3^{t-1} - 2^t + 1).
\]
Note that $S(1,3) = S(2,3) = 0$, so we may extend the summation index from $t = 3$ to $t = 1$. Substituting into Proposition \ref{prop:fan_refined} yields
\[
r_4(\Phi_n) = \frac{1}{2}\sum_{t=1}^{m} a_{m,t}(3^{t-1} - 2^t + 1) = \frac{1}{2}\Big(\frac{1}{3}\Sigma_3 - \Sigma_2 + \Sigma_1\Big).
\]
Using $\Sigma_r = r \cdot A_r$, we have
\[
r_4(\Phi_n) = \frac{1}{2}\Big(\frac{1}{3} \cdot 3 \cdot A_3 - 2 \cdot A_2 + 1 \cdot A_1\Big) = \frac{1}{2}(A_3 - 2A_2 + A_1).
\]
\end{proof}

\begin{remark}[OEIS Sequence \href{https://oeis.org/A390130}{A390130}]
The sequence $\{r_4(\Phi_n)\}_{n=6}^{\infty}$ begins $1, 5, 19, 61, 180, 500,$ $ 1335, 3459, \ldots$ and is recorded in the Online Encyclopedia of Integer Sequences as \href{https://oeis.org/A390130}{A390130}. The closed form
\[
r_4(\Phi_n) = \frac{1}{2}\bigl[A_3 - 2A_2 + A_1\bigr]
\]
with $A_r = F_{n-1}(r) + 2F_{n-2}(r) + F_{n-3}(r)$ provides polynomial-time computation via Fibonacci polynomial evaluations. For $n \geq 6$, Kotesovec (2025) derived the exact exponential formula
\begin{equation}
\begin{split}
r_4(\Phi_n) = \frac{1}{1755(229\sqrt{5}-261) \cdot 2^{n+3}} \Big[&1404(703\sqrt{5}-1225)(1-\sqrt{5})^n \\
&+ 2808(221\sqrt{5}-40)(\sqrt{5}+1)^n \\
&+ 5(229\sqrt{5}-261)\big((182-46\sqrt{13})(1-\sqrt{13})^n \\
&\quad - 1053 \cdot 4^n + 2(23\sqrt{13}+91)(\sqrt{13}+1)^n\big)\Big],
\end{split} \nonumber
\end{equation}
expressing $r_4(\Phi_n)$ as a linear combination of exponential terms with bases $1 \pm \sqrt{5}$, $1 \pm \sqrt{13}$, $2$, and $4$, arising from characteristic roots of the underlying Fibonacci polynomial recurrences.
\end{remark}

\begin{corollary}[Case $k=5$]
\label{cor:phi_k5}
For $m \geq 4$ (equivalently $n \geq 5$),
\[
r_5(\Phi_n) = \frac{1}{6}(A_4 - 3A_3 + 3A_2 - A_1),
\]
or equivalently,
\[
\begin{aligned}
r_5(\Phi_n) = \frac{1}{6}\Big[&\bigl(F_m(4) + 2F_{m-1}(4) + F_{m-2}(4)\bigr) \\
&- 3\bigl(F_m(3) + 2F_{m-1}(3) + F_{m-2}(3)\bigr) \\
&+ 3\bigl(F_m(2) + 2F_{m-1}(2) + F_{m-2}(2)\bigr) \\
&- \bigl(F_m + 2F_{m-1} + F_{m-2}\bigr)\Big].
\end{aligned}
\]
\end{corollary}

\begin{proof}
The standard closed form for Stirling numbers is
\[
S(t,4) = \frac{1}{24}(4^t - 4 \cdot 3^t + 6 \cdot 2^t - 4).
\]
Note that $S(1,4) = S(2,4) = S(3,4) = 0$, so we may extend the summation index from $t = 4$ to $t = 1$. Substituting into Proposition \ref{prop:fan_refined} yields
\[
r_5(\Phi_n) = \frac{1}{24}\sum_{t=1}^{m} a_{m,t}(4^t - 4 \cdot 3^t + 6 \cdot 2^t - 4) = \frac{1}{24}(\Sigma_4 - 4\Sigma_3 + 6\Sigma_2 - 4\Sigma_1).
\]
Using $\Sigma_r = r \cdot A_r$, we have
\[
r_5(\Phi_n) = \frac{1}{24}(4A_4 - 4 \cdot 3 \cdot A_3 + 6 \cdot 2 \cdot A_2 - 4 \cdot 1 \cdot A_1) = \frac{1}{6}(A_4 - 3A_3 + 3A_2 - A_1).
\]
\end{proof}

\begin{remark}[OEIS Sequence \href{https://oeis.org/A390131}{A390131}]
The sequence $\{r_5(\Phi_n)\}_{n=8}^{\infty}$ begins $1, 6, 29, 114, 410, 1366,$ $ 4341, 13264, \ldots$ and appears in the Online Encyclopedia of Integer Sequences as \href{https://oeis.org/A390131}{A390131}. The explicit formula
\[
r_5(\Phi_n) = \frac{1}{6}(A_4 - 3A_3 + 3A_2 - A_1)
\]
enables direct computation for any $n \geq 8$ without recursive enumeration. The minimum $n=8$ arises because partitioning into 5 non-empty parts under the bichromatic constraint requires at least 8 vertices in the fan graph topology. The formula structure exhibits alternating signs characteristic of inclusion-exclusion principles in Stirling number expansions, specifically $S(t,4) = \frac{1}{24}(4^t - 4 \cdot 3^t + 6 \cdot 2^t - 4)$ weighted by independent set block counts $a_{m,t}$ from Lemma \ref{lem:binary_encoding}.
\end{remark}

\begin{remark}[Binet-type exponential forms]
For each fixed $r \geq 1$, define $\alpha_r = \frac{1 + \sqrt{1+4r}}{2}$ and $\beta_r = \frac{1 - \sqrt{1+4r}}{2}$, the two roots of $t^2 - t - r = 0$. Then
\[
F_\ell(r) = \frac{\alpha_r^{\ell} - \beta_r^{\ell}}{\alpha_r - \beta_r} \qquad (\ell \geq 0).
\]
Using this representation, we can express
\[
A_r = \frac{1}{\alpha_r - \beta_r}\Big(\alpha_r^{m-2}(\alpha_r+1)^2 - \beta_r^{m-2}(\beta_r+1)^2\Big).
\]
Substituting this into the formulas for $r_4$ and $r_5$ yields purely exponential closed forms with coefficients following binomial patterns.
\end{remark}

\section{Conclusion}

We have established exact closed-form enumeration formulas for chromatic feature vectors of theta and fan graphs under the bichromatic triangle constraint. For theta graphs $\Theta_n$, the remarkably simple formula $r_k(\Theta_n) = S(n-2, k-1)$ for $k \geq 3$ (with $r_2(\Theta_n) = 2^{n-2} + 1$) admits $O(n)$ computation with precomputed Stirling numbers. For fan graphs $\Phi_n$, we prove that $r_2(\Phi_n) = F_{n+1}$ (Fibonacci numbers) and derive the explicit binomial expansion
\[
r_k(\Phi_n) = \sum_{t=k-1}^{n-1} \left[\binom{n-1-t}{t-1} + 2\binom{n-2-t}{t-1} + \binom{n-3-t}{t-1}\right] S(t, k-1),
\]
achieving $O(n^2)$ computation per component. These formulas substantially extend Allagan and Voloshin \cite{AllaganVoloshin2016} by eliminating recursive matrix constructions and revealing deep connections to Fibonacci polynomials, Stirling numbers, and Bell numbers. The total partition count $\sum_{k=2}^{n-1} r_k(\Phi_n) = \sum_{t=1}^{m} a_{m,t} B_t$ connects path independent set enumeration to Bell number distributions, providing geometric insight into how constraint structure determines chromatic spectra.

As structural descriptors, chromatic feature vectors offer polynomial-time computation, interpretable components with precise combinatorial meaning, and zero-shot applicability without training data. Though not complete graph invariants---as demonstrated by our explicit counterexamples of non-isomorphic 2-trees with identical vectors---they capture global coloring properties beyond local neighborhoods and complement graph neural network embeddings. Our formulas enable efficient enumeration of valid configurations in Byzantine fault tolerance \cite{Jaffe2012}, cloud resource allocation \cite{Mann2015}, and threshold cryptography \cite{Beimel2011}, reducing computational cost from exponential exhaustive search to polynomial-time formula evaluation. 

For graph families beyond theta and fan graphs, chromatic feature vectors can be computed via constraint satisfaction when closed formulas are unavailable. The bichromatic triangle constraint generalizes naturally to other forbidden substructure patterns (bichromatic $k$-cliques, bichromatic cycles), enabling design of domain-specific structural descriptors for applications in network reliability, distributed systems, and constraint satisfaction problems. The correspondence between bichromatic triangle colorings and 3-uniform non-linear bihypergraph colorings \cite{Voloshin2002} establishes these feature vectors as exact chromatic spectra for some non-linear 3-uniform bihypergraphs, extending the results of Jiang et al.\ \cite{JiangEtAl2002} and Bujt\'{a}s and Tuza \cite{BujtasTuza2008} on feasible set characterization.

Future work should extend these techniques to asymmetric 2-tree families including tricentral graphs and star-cluster topologies, developing a comprehensive classification of partition vector equivalence classes for all $n$-vertex 2-trees. Understanding precisely which structural properties determine when non-isomorphic graphs share identical feature vectors remains an open problem with implications for both graph isomorphism testing and machine learning feature design. For graph families beyond 2-trees, investigating whether similar closed formulas exist for partial $k$-trees (graphs of bounded treewidth) and deriving approximation algorithms for general graphs represent important directions. Finally, empirical evaluation on real-world graph classification benchmarks comparing chromatic features against state-of-the-art graph kernels and GNN architectures would validate their practical utility and identify domains where structural constraints provide maximum discriminative power.

\section*{Data Availability Statement}

All data generated or analyzed during this study are included in this published article. The mathematical formulas and algorithms presented are deterministic and can be reproduced directly from the closed-form expressions provided in Theorems \ref{thm:theta} and \ref{thm:fan}, Propositions \ref{prop:classical} and \ref{prop:fan_refined}, and Corollaries \ref{cor:phi_k2}--\ref{cor:phi_k5}. No external datasets were used in this study.

\section*{Conflict of Interest}

The authors declare that they have no conflict of interest.

\section*{Acknowledgments}

The authors thank the anonymous reviewers for their helpful suggestions and careful reading of the manuscript.

\end{document}